\documentclass[aps,pra,english,10pt,a4,nofootinbib,notitlepage,superscriptaddress,twocolumn]{revtex4-1}
\usepackage{hyperref,graphicx, amsmath, amssymb, amsthm, color, bm, bbm,dsfont}
\usepackage{cleveref}

\theoremstyle{plain}% default
\newtheorem{thm}{Theorem}
\newtheorem{lem}[thm]{Lemma}
\newtheorem{prop}[thm]{Proposition}

\theoremstyle{definition}
\newtheorem{defn}[thm]{Definition}
\newtheorem{conj}[thm]{Conjecture}

\definecolor{nblue}{rgb}{0.2,0.2,0.7}
\definecolor{ngreen}{rgb}{0.1,0.5,0.1}%272
\definecolor{nred}{rgb}{0.8,0.2,0.2}%711
\definecolor{nblack}{rgb}{0,0,0}

\DeclareMathOperator{\tr}{\mathrm{Tr}}

\newcommand{\mc}[1]{\mathcal{#1}}

\newcommand{\bc}[1]{\bm{\mathcal{#1}}}
\newcommand{\mbb}[1]{\mathbb{#1}}

\newcommand{\inner}[2]{\left\langle #1 ,#2 \right\rangle}

\newcommand{\ct}{^{\dagger}}

\newcommand{\hidden}[1]{}

\usepackage{siunitx}
\usepackage{amsmath}
\usepackage{siunitx}
\usepackage{amsmath}
\usepackage{graphicx}
\usepackage{cancel}
\usepackage{caption}
\usepackage{subcaption}
\usepackage[T1]{fontenc}
\DeclareMathOperator*{\argmax}{argmax}

\usepackage{siunitx}
\begin{document}
\title{
	From randomized benchmarking experiments to gate-set circuit fidelity:
	\\
	{{how to interpret randomized benchmarking decay parameters}}
}
	\author{Arnaud Carignan-Dugas}
	\email{arnaud.carignan@gmail.com}
	\affiliation{Institute for Quantum Computing and the Department of Applied
		Mathematics, University of Waterloo, Waterloo, Ontario N2L 3G1, Canada}
	\author{Kristine Boone}
	\affiliation{Institute for Quantum Computing and the Department of Applied
			Mathematics, University of Waterloo, Waterloo, Ontario N2L 3G1, Canada}
	\author{Joel J. Wallman}
	\affiliation{Institute for Quantum Computing and the Department of Applied
		Mathematics, University of Waterloo, Waterloo, Ontario N2L 3G1, Canada}
	\author{Joseph Emerson}
	\affiliation{Institute for Quantum Computing and the Department of Applied
		Mathematics, University of Waterloo, Waterloo, Ontario N2L 3G1, Canada}
	\affiliation{Canadian Institute for Advanced Research, Toronto, Ontario M5G 1Z8, Canada}
\date{\today}
\begin{abstract}

Randomized benchmarking (RB) protocols have become an essential tool for
providing a meaningful partial characterization of experimental quantum
operations. {While the RB decay rate is known to enable estimates of the average
fidelity of those operations under gate-independent Markovian noise, under gate-dependent noise this rate is more difficult to interpret rigorously.} In this paper, we prove that {single-qubit RB decay parameter $p$ coincides with the decay parameter of the}{ \emph{gate-set circuit fidelity}}, a
novel figure of merit which characterizes the expected average fidelity over
arbitrary circuits of operations from the gate-set. We also prove that, in the
limit of high-fidelity single-qubit experiments, 
{the possible alarming disconnect between the average gate fidelity and RB experimental results
is simply explained by} a basis mismatch between the gates and
the state-preparation and measurement procedures, that is, to a unitary degree
of freedom in labeling the Pauli matrices. Based on numerical evidence and
physically motivated arguments, we conjecture that these results also hold for higher
dimensions.

\end{abstract}

\maketitle

\section{Introduction}

The operational richness of quantum mechanics hints at an unprecedented
computational power. {However, this very richness carries over to a vast range of
possible quantum error processes for which a full characterization is impractical for even a handful of quantum bits
(qubits).} Randomized
benchmarking (RB) experiments \cite{Emerson2005,Levi2007, Knill2008,Dankert2009, Magesan2011, Magesan2012a,Dugas2015,Cross2016} were introduced to provide
a robust, efficient, scalable, SPAM-independent\footnote{SPAM stands for ``State preparation and measurement''.}, partial characterization of specific sets of quantum operations of interest, referred to as gate-sets. Such experiments have been widely adopted across all platforms for quantum computing, eg. \cite{Gaebler2012,Corcoles2013,Kelly2014,Barends2014,Casparis2016,Takita2016,Sheldon2016,McKay2016,McKay2017}, and have become a critical tool for characterizing and improving the design and control of quantum bits (qubits).

Recently it has been shown that RB experiments on an arbitrarily large number of qubits will always produce an exponential decay under arbitrary Markovian error models {(that is, where errors are represented as completely-positive maps)}. This ensures a well-defined theoretical characterization of these experiments and enables an important test for the presence of non-Markovian errors, in spite of the gauge freedom between the experimental quantities and a theoretical figure of merit such as the average gate fidelity \cite{Proctor2017,Wallman2017, Merkel2018}. However, this theoretical advance still lacks a clear physical interpretation that
rigorously connects the experimentally observed decay to
a fidelity-based characterization of a physical set of gate-
dependent errors. {Linking an experimentally measured quantity to a physically meaningful figure of merit is not a mere intellectual satisfaction. It is crucial to ensure that a quantity measured in the context of a process characterization protocol indeed yields an outcome that assesses the quality of 
	operations. What if a very noisy quantum device could yield a decent RB parameter?  What if there exist scenarios where RB substantially underestimates the quality of a quantum device?}

In this paper, we show that in the regime of high fidelity gates on single qubits, {a simple physical} interpretation {of RB data} does exist {and allows a reliable characterization of quantum operations}. 
 Further we conjecture, based on numerical evidence, that such an interpretation extends to arbitrary dimensions. {Consequently, this work provides the theoretical foundation behind a fundamental tool for identifying and eliminating errors through examining the results of RB experiments.}

Consider a {targeted} ideal gate-set $\mbb G=\{\mc G\}$ and its noisy implementation $\tilde{\mbb G}=\{\tilde{\mc G}\}$. We denote a circuit composed of $m$ elements by
\begin{align}
\tilde{\mc G}_{m:1}:=
\tilde{\mc G}_m \cdots \tilde{\mc G}_2 \tilde{\mc G}_1~.
\end{align}
For leakage-free RB experiments with arbitrarily gate-dependent (but still Markovian) errors, the average probability of an outcome $\mu$ after preparing a state $\rho$ and applying a circuit of $m+1$ operations that multiply to the identity is~\cite{Wallman2017,Merkel2018}
\begin{align}\label{eq:intrinsic_RB}
\mbb E_{\mc G_{m+1:1}}{\inner{\mu}{\tilde{\mc G}_{m+1:1} (\rho)}}
= A p^m+B+\epsilon(m),
\end{align}
{where $\inner{M_1}{M_2}:= \tr {M_1^\dagger M_2}$ refers to the Hilbert-Schmidt inner product. On the right-hand side of \cref{eq:intrinsic_RB},} $A$ and $B$ are independent of $m$ (i.e., they only depend upon $\rho$,
$\mu $ and $\tilde{\mc G}$) and $\epsilon(m)$ is a perturbative term that decays
exponentially in $m$.

By design, RB gives some information about the error rate of motion-reversal (i.e.,
identity) circuits composed of gate-set elements. In this paper, we show how this information relates to general
circuits. Consider the traditional notion of \emph{average fidelity} for a noisy circuit $\tilde{\mc C}$ to
a {target} unitary circuit $\mc C$,
\begin{align}\label{eq:average_circuit_fidelity}
{\rm F}(\tilde{\mc C}, \mc C ) & :=  \int {\inner{\tilde{\mc C}(\psi)}{\mc C(\psi)}}d \psi,
\end{align}
where the integral is taken uniformly over all pure states.
%Similarly, we denote the circuit average infidelity as $r(\tilde{\mc C}, \mc C)=1-{\rm F}(\tilde{\mc C}, \mc C)$.
\Cref{eq:average_circuit_fidelity} corresponds to the definition of the usual notion of \emph{average gate fidelity}, but is instead formulated in terms of ``circuit'', which is to be understood as a sequence of elementary gates.
We introduce this nuance to define a novel figure of merit, the \emph{gate-set circuit fidelity}, which compares all possible sequences of $m$ implemented operations from the gate-set $\tilde{\mbb G}$ to their {targets} in $\mbb G$,
\begin{defn}[{\emph{gate-set circuit fidelity}}]
	\begin{align}\label{eq:avg_circuit_fid_construct}
	\mc F(\tilde{\mbb G}, \mbb G, m):=
	\mbb E \left[{\rm F}(\tilde{\mc G}_{m:1}, {\mc G}_{m:1} )\right]~.
	\end{align}
\end{defn}
The case $m=1$ yields the average fidelity of the
gate-set $\tilde{\mbb G}$ to ${\mbb G}$.
In general, the overall action of ideal circuits
${\mc G}_{m:1}$ is reproduced
by
$ \tilde{\mc G}_{m:1}$
with fidelity $\mc F (\tilde{\mbb G},\mbb G, m)$. Having access to the gate-set circuit fidelity enables going beyond quantifying the quality of gate-set elements as it also quantifies the quality of circuits based on those elements.
In this paper, we prove that for all single-qubit gate-sets with fidelities
close to 1 {and for an appropriately chosen targeted gate-set $\mbb G$}, the gate-set circuit fidelity can be closely  estimated via RB
experiments, for all circuit lengths $m$, {even in cases of highly gate-dependent noise models}. {This is possible because it turns out that $\mc F(\tilde{\mbb G}, \mbb G, m)$ essentially behaves like an exponential decay in $m$, uniquely parametrized by the RB decay constant $p$.} {The robust inclusion of gate-dependence is a major step forward since
	it encompasses very realistic noise models.} We conjecture this result to hold for higher dimensions, based
on numerical evidences and physically motivated arguments. 
 {Notice that the gate-set circuit fidelity
 	quantifies the expected fidelity of \emph{all} circuits (built from gate-set elements), and not only motion-reversal ones. This is an important observation to keep in mind because although RB experiments intrinsically revolve around motion-reversal circuits, the figure of merit that it yields isn't limited to such paradigm. Quantifying the quality of all 
 	circuits is much more useful than quantifying identity ones.}
\section{The dynamics of the gate-set circuit fidelity}\label{sec:RB_marginal}

It follows from the RB literature \cite{Emerson2005,Magesan2011} that for gate-independent noise models of the form $\tilde{\mbb G}=\mc E \mbb G$ or ${\tilde{\mbb G}= \mbb G \mc E}$, where $\mc E$ is a fixed error, the gate-set circuit fidelity behaves exactly as
\begin{align}\label{eq:gate-set_fid_gate_indep}
\mc F(\tilde{\mbb G},\mbb G, m)=\frac{1}{d}+\frac{d-1}{d}p^m~,
\end{align}
where $p$ is estimated through standard RB by fitting to \cref{eq:intrinsic_RB}
with $\epsilon(m) = 0$
and $d$ is the dimension of the system. {The relationship between the survival probability decay curve and the decay in \cref{eq:gate-set_fid_gate_indep} shouldn't be surprising. Indeed, consider a RB experiment with a noise model of the form 
$\mc E \mbb G$ and a perfect inversion
step $\mc G_{m+1} \in \mbb G$ and perfect SPAM. In such case, the gate-set circuit fidelity
and the survival probability exactly coincide. A less trivial matter is to show
the link between the RB decay parameter and \cref{eq:gate-set_fid_gate_indep}
for gate-dependent
leakage-free noise models for which the choice of targeted gate-set is
to be treated more carefully. In fact, as we will show,
a poor choice of targeted gate-set can lead to 
a strong violation of \cref{eq:gate-set_fid_gate_indep} in the sense that $1-\mc F(\tilde{\mbb G},\mbb G, m)$ can relatively differ from $1-(\frac{1}{d}+\frac{d-1}{d}p^m)$ by multiple orders of magnitude. {An appropriate choice of targeted gate-set will essentially restore
	the decay relation shown in \cref{eq:gate-set_fid_gate_indep}.}

\Cref{eq:gate-set_fid_gate_indep} generalizes to
\begin{align}\label{eq:gate-set_fid_gate_dep}
\mc F(\tilde{\mbb G},\mbb G, m)=\frac{1}{d}+\frac{d-1}{d}f_{\rm tr}(\tilde{\mbb G}, \mbb G, m)~,
\end{align}
where the fidelity on the traceless hyperplane is similar to the
gate-set circuit fidelity, but is
averaged over the traceless part of the pure states, $\psi_{\rm tr}= \psi-~\mbb I /d$:
\begin{align}\label{eq:Bloch_fid}
 f_{\rm tr} (\tilde{\mbb G} ,\mbb G, m) :=
\text{\scalebox{0.9}{$\frac{\mbb E \left( \int {\inner{\tilde{\mc G}_{m:1}(\psi_{\rm tr})}{
	{\mc G}_{m:1} (\psi_{\rm tr})}}  d \psi \right)}
{\int {\inner{\psi_{\rm tr}}{\psi_{\rm tr}}} d \psi }$} .}
\end{align}
{The integrand in the numerator of the right-hand side of \cref{eq:Bloch_fid} can be visualized as the fidelity restricted on the Bloch space, comparing the
ideally mapped Bloch vectors $\psi_{\rm tr} \rightarrow {\mc G}_{m:1} (\psi_{\rm tr})$ to their noisy analog $\tilde{\mc G}_{m:1} (\psi_{\rm tr})$. \Cref{eq:gate-set_fid_gate_dep}
is quickly obtained by realizing that the symmetric integral over the Bloch space $\int \psi_{\rm tr} d \psi =0$. }

{Under gate-dependent noise, $1-f_{\rm tr}(\tilde{\mbb G},\mbb G,1)$ could
 relatively differ from $1-p$~ by several orders of magnitude \cite{Proctor2017,Qi2018}. Such discrepancy was
  seen as a serious concern: the observed RB decay
 seemingly fails in characterizing the quality of quantum operations.
 To see the possible immense disconnect between $p$ and $f_{\rm tr}(\tilde{\mbb G},\mbb G,1)$, consider the canonical example where single-qubit gates
are perfectly implemented, but differ from the targets $\mc G \in \mbb G$
by a labeling of the Pauli axes:
\begin{subequations}
	\begin{align}
	\tilde{\mc G}(X)&= \mc G (Y)~, \label{eq:modelA}\\
	\tilde{\mc G}(Y)&= \mc G (Z)~, \label{eq:modelB} \\
	\tilde{\mc G}(Z)&= \mc G (X)~. \label{eq:modelC}
	\end{align} 
\end{subequations}
This noise model would lead to an abscence of decay in the survival
probability, that is $p=1$. Indeed, motion-reversal circuits
are perfectly implementing the identity gate, regardless of the length of the circuit.
A quick calculation results in $f_{\rm tr}(\tilde{\mbb G},\mbb G,m)=0$, which demonstrates a difference in orders of magnitude $|\log(1-p)-\log(1-f_{\rm tr}(\tilde{\mbb G},\mbb G,1))|$ that tends to infinity as $p \rightarrow 1$.
 The RB experiment indicates no operational error while the average gate fidelity indicates $1/2$. Does the outcome of RB massively underestimate the error? Notice that since the implementation
 error is a permutation of labels, there is actually
 no observable error in the device. The alarmingly low value of gate-set circuit fidelity of $\tilde{\mbb G}$ to $\mbb G$
 is simply a consequence of a poor choice of targeted gate-set. 
}
\\
\\
 {As a more involved example}, let {the noise model be} $\tilde{\mbb
G}= \mc U \mbb G \mc U ^\dagger$ for any non-identity unitary channel $\mc
U$ {and let the set of targeted operations be $\mbb G$} {(this includes our previous mislabeling example as a special scenario)}. In such cases $f_{\rm tr}(\tilde{\mbb G},\mbb G, 1)$ {can take any value in the interval $[0,1)$, depending on the choice of $\mc U$.} {However, using the same argument as in the previous example,
the survival probability is not subject to a decay ($p=1$),
showing once again how the decay parameter could arbitrarily differ from a poorly defined average gate fidelity.
This apparent disconnect
arises due to a \textit{basis mismatch} between the bases in which
the noisy gate-set and the targeted gate-set are defined. A reconciliation of the RB observations with a gate-set circuit fidelity is obtained by changing the set of targets to $\mc U \mbb G \mc U^\dagger$ since $f_{\rm tr}(\tilde{\mbb G},\mc U \mbb G \mc U^\dagger, 1) = 1$. One might argue that 
implementing $\tilde{\mbb G} = \mc U \mbb G \mc U^\dagger$ 
instead of the ideal $\mbb G$
should raise an operational error. Not necessarily: consider a circuit uniquely constructed from operations $\tilde{\mc G}_i \in \tilde{\mbb G}$. According to Born's rule, the probability of measuring the outcome $i$ associated with the positive operator $\mu_i$ after performing the circuit on a state $\rho$ is:
\begin{align}
	p_i&= \inner{\mu_i}{\tilde{\mc G}_{m:1} (\rho)} \notag \\
	   &= \inner{\mu_i}{\mc U \mc G_m \mc U^\dagger \cdots \mc U \mc G_2 \mc U^\dagger \mc U \mc G_1 \mc U^\dagger(\rho)} \notag \\
	   &= \inner{\mu_i}{\mc U \mc G_{m:1} \mc U^\dagger(\rho)} \notag \\
	   &= \inner{\mu_i'}{ \mc G_{m:1} (\rho')}~,
\end{align}
where $\rho'= \mc U^\dagger(\rho)$, $\mu_i'= \mc U^\dagger(\mu_i)$. That is, the error can be interpreted as part of SPAM procedures instead of operations. {Since the unitary transformation can be pushed to either SPAM procedures or coherent manipulations, it should be
	seen as a mismatch between them. Indeed, the physical unitary conjugation $\tilde{\mbb G} = \mc U \mbb G \mc U^\dagger$ doesn't affect the \emph{internal action} of operations, but exclusively the connection between operations and SPAM procedures.} Changing the targeted gate-set $\mbb G$ to $\mc U \mbb G \mc U^\dagger$ is allowed by the degree of freedom in labeling what is the basis for SPAM procedures
and what is the basis for processes. 

In \cref{sec:appendixA}, we {show how exactly} the disconnect between $p^m$ and
$f_{\rm tr}(\tilde{\mbb G},\mc U \mbb G \mc U^\dagger, m)$ depends on the choice of {targeted gate-set $\mc U \mbb G \mc U^\dagger$}. 
That is, we {provide an expression of the form}
\begin{align}\label{eq:Bloch_total}
		& f_{\rm tr} (\tilde{\mbb G} ,\mc U \mbb G \mc U^\dagger, m)
		= C(\mc U) p^m + D(m, \mc U)~,
\end{align}
where $\mc U$ is a physical unitary channel (see \cref{thm:total_change}).
{A first interesting observation is that} $D(m, \mc U)$ is typically negligible or becomes rapidly negligible as it is also
exponentially suppressed in $m$\footnote{Since $D(1, \mc U)$ is typically close to $0$, the exponential suppression is quite effective compared to $p^m \approx 1- m(1-p)$ which is essentially linear for small $m$.}. This means that the relative variation
in $f_{\rm tr}$ as the circuit grows in length,
\begin{align}\label{eq:marginal_law}
\frac{ f_{\rm tr}(\tilde{\mbb G}, \mc U \mbb G \mc U^\dagger, m+1)}
{ f_{\rm tr}(\tilde{\mbb G}, \mc U \mbb G \mc U^\dagger, m)}
&= p+\delta(m, \mc U)~,
\end{align}
depends weakly on the choice of {targeted gate-set}.
More precisely, $\delta(m, \mc U)$ is composed of two factors:
the first one decays exponentially in $m$ and is at most of order $(1-p)^{m/2}$,
while the second carries the dependence in $\mc U$; the existence of
a specific choice of $\mc U$ such that this last factor
becomes at most of order $(1-p)^{3/2}$ is
proven in the single-qubit case (\cref{sec:appendixB}),
and conjectured to hold in general. The explicit behaviour of $\delta (m,\mc U)$
given a numerically simulated gate-dependent noise model
is illustrated in \cref{fig:delta}.

Consequently, the gate-set circuit fidelity can be updated with a good approximation through the recursion relation
\begin{equation}\label{eq:update_fid}
\mc F (\tilde{\mbb G} ,\mc U\mbb G \mc U\ct, m+1) \approx
\frac{1}{d}+p \left(\mc F(\tilde{\mbb G} ,\mc U\mbb G \mc U\ct, m)-\frac{1}{d}\right).
\end{equation}
{Roughly speaking, this means that the choice of basis $\mc U$ in which are expressed the targets in $\mc U \mbb G \mc U^\dagger$ is not highly significant when it comes to updating the gate-set circuit fidelity as the circuit grows in depth.} {The RB decay rate $p$ enables the decrease in fidelity due to adding a gate to a circuit to be predicted.}

{However,} to provide insight on the total value of the gate-set circuit fidelity
given a circuit's length $m$, we need a stronger relation between
the RB estimate of $p$ and the gate-set circuit fidelity.
Fortunately,
the {basis} freedom in the choice of {targeted} gate-set can be fixed
in a way that allows us to estimate the total change in gate-set
circuit fidelity
for arbitrary circuit's lengths.

In \cref{sec:appendixB}, we prove that the {potentially large} disconnect between $p$ and
$f_{\rm tr}(\tilde{\mbb G},\mc U \mbb G \mc U^\dagger, 1)$ under general gate-dependent noise is almost completely accounted for by a basis mismatch {which, as we argued earlier, doesn't exactly correspond to
	a process error since unitary conjugation does not affect the internal dynamics of operations.} 
\begin{prop}
	\label{prop:main}
	For any single-qubit noisy gate-set $\tilde{\mbb G}$ perturbed from $\mbb G$, there exists an ideal {targeted} gate-set $\mc U \mbb G  \mc U^\dagger$, where $\mc U$ is a physical unitary, such that
	\begin{align}\label{eq:main}
\mc F (\tilde{\mbb G} ,\mc U \mbb G \mc U^\dagger , m) = \frac{1}{d}+\frac{d-1}{d} p^m + O\left((1-p)^2\right)~.
	\end{align}
\end{prop}

In fact, we conjecture this result to hold for any dimension,
or at least for most realistic gate-dependent noise models.
To grasp the physical reasoning behind this,
we refer to the end of \cref{sec:appendixB},
as it rests on some prior technical analysis. The extension of
\cref{prop:main} to $2$-qubit systems is
supported by numerical evidences (see \cref{fig: pbloch2,fig:delta2,fig:basis_adjustment2}).

{The unitary freedom appearing in the gate-set circuit fidelity 
	means that there exists an infinite amount of fidelity-based figures
	of merit describing noisy circuits, one for each infinitely many targeted gate-set $\mc U \mbb G \mc U^\dagger$. Of course, there exist choices of targeted operations that yield in gate-set circuit fidelities
	that differ from \cref{{eq:main}} (see \cite{Proctor2017,Qi2018}); the example
	shown in \cref{eq:modelA,eq:modelB,eq:modelC} is an elementary instance thereof.
	\Cref{prop:main} simply states that there exists a natural choice of gate-set $\mc U \mbb G \mc U^\dagger$ that allows connecting the outcome of an RB experiment to a gate-set circuit fidelity. The choice of
	basis $\mc U$ is like taking the perspective of the gates rather than the perspective of SPAM procedures (as is implicitly done when defining gates relative to the energy eigenbasis of the system). In this picture, the gate-set circuit fidelity describes the accuracy of the internal behaviour of operations as they act in concert.
	}

To reformulate the result, the family of circuits
$\tilde{\mc G}_{m:1}$ built from
a composition of $m$ noisy operations
$\tilde{\mc G} \in \tilde{\mbb G}$ mimics
the family of ideal circuits $\mc U \mc G_{m:1} \mc U^\dagger$
with fidelity $\frac{1}{d}+\frac{d-1}{d}p^m$. In the paradigm where
the {initially targeted} operations $\mc G \in \mbb G$ are defined with respect to SPAM
procedures,
$\mc U$ captures the misalignment between
the basis in which the operations $\tilde{\mc G} \in \tilde{\mbb G}$ are defined and the
basis defined by SPAM procedures. This goes farther: consider an
additional gate-set, for which the {targeted} operations $\mc H \in \mbb H$
are also are defined respect to SPAM
procedures. From \cref{prop:main}, there exists a physical unitary
$\mc V$ for which $\tilde{\mc H}_{m:1}$ imitates the action of $\mc V \mc H_{m:1} \mc V^\dagger$ with fidelity
$\frac{1}{d}+\frac{d-1}{d}q^m$ (where $q$ is estimated through RB). $\mc U^\dagger \mc V$ captures the basis mismatch between the gate-sets $\tilde{\mbb G}$ and $\tilde{\mbb H}$. Such a non-trivial mismatch could easily be imagined if, for instance, gates
belonging to $\tilde{\mbb H}$ were obtained through a different physical process than $\tilde{\mbb G}$,
or calibrated with regards to alternate points of reference.

{\section{Finding the appropriate set of targeted gates for specific noise models}}\label{sec:operational}

We now discuss how the {appropriate unitary conjugation on the initial targeted gate-set} can be calculated for specific noise
models, whether from numerical simulations, analytic approximations, or
tomographic reconstructions.
As shown in \cref{thm:total_change} and \cref{eq:Bloch_total},
the total change of gate-set circuit fidelity depends on the physical basis in
which the ideal gate-set is expressed. In the single-qubit case, we showed the
existence of a physical basis $\mc U$ that reconciles
$f_{\rm tr} (\tilde{\mbb G} ,\mc U \mbb G \mc U^\dagger, m)$
with $p^m$ through \cref{prop:main}. One might suspect that the
unitary $\mc U$ can be found through the maximization of the gate-set fidelity:
\begin{align}\label{eq:wrong_hypothesis}
	\mc U= \argmax\limits_{\mc V} \mc F (\tilde{\mbb G} ,\mc V \mbb G \mc V^\dagger, 1)~,
\end{align}
and indeed this would handle noise models of the form
${\tilde{\mbb G}= \mc U \mc E  \mbb G \mc U^\dagger}$, as
\begin{align*}
p &= f_{\rm tr}(\tilde{\mbb G}, \mc U \mbb G \mc U^\dagger, 1) \geq f_{\rm tr}(\tilde{\mbb G}, \mbb G, 1)~.
\end{align*}
However, this hypothesis fails for simple noise models of the form
{$\tilde{\mbb G}= \mc U^\dagger \mc E \mbb G \mc U^\dagger$, where
\begin{align*}
p &= f_{\rm tr}(\tilde{\mbb G}, \mc U \mbb G \mc U^\dagger, 1)  \leq f_{\rm tr}(\tilde{\mbb G}  , \mbb G, 1)~.
\end{align*}
Those last two examples show that $p$ can be greater or less than $f_{\rm tr}(\tilde{\mbb G}  , \mbb G, 1)$, depending on the noise model. More examples are derived in \cite{Proctor2017,Qi2018}.}
This particular case study is informative as these two last noise models
share something in common:
there exists a choice of unitary {$\mc U$} that cancels the noisy map
on the right of the noisy gate-set.
Although such exact cancellation is not always possible, we now show that a close approximation is sufficient. Consider the slightly more general noise model of the form
$\tilde{\mbb G} = \mc E_L \mbb G \mc E_R$, where we allow fixed
but arbitrary error maps to the left and the right of an ideal gate-set. {It can be shown while staying under the scope of the original analysis provided in \cite{Magesan2011, Magesan2012a}} that $p^m= f_{\rm tr}(\mc E_R \mc E_L \mbb G, \mbb G, m)$, since $\mc E_R \mc E_L$ is the effective
error map between two otherwise perfect implementations of the
gate-set elements. In the single-qubit case
(and for many, if not all physically motivated higher dimensional noise models)
there exists a unitary operation $\mc U$ such that
\begin{equation}\label{eq:match_fid}
	{\rm F}(\mc E_R \mc E_L, ~ \mc I) = \mc F(\mc E_L \mbb G \mc E_R, ~ \mc U \mbb G \mc U^\dagger, 1) + O((1-p)^2)~,
\end{equation}
(see \cref{sec:appendixB}). That is, the fidelity of the map between two noisy gate-sets can be seen as the gate-set circuit fidelity between a noisy gate-set and an {appropriately targeted ideal} one. A choice of such physical
unitary is
\begin{align}
\mc U = \argmax\limits_{\mc V} F\left(\mc E_R \mc V, \mc I\right)~,
\end{align}
which essentially cancels the unitary part of $\mc E_R$ \footnote{Of course,
$\argmax\limits_{\mc V} F\left( \mc V^{\dagger} \mc E_L , \mc I\right)$
would also fulfill \cref{eq:match_fid}.}. {Another way to see this is that the unitary freedom allows us to reexpress the errors $\mc E_L,~\mc E_R$ as
\begin{align*}
	\mc E_L &\rightarrow \mc U^\dagger \mc E_L  \\
	\mc E_R &\rightarrow \mc E_R \mc U~.
\end{align*}
We can then chose the unitary that depletes $\mc E_R \mc U$ from any coherent component.
Intuitively, reexpressing the error on one side to make it incoherent prevents any type of unitary conjugation
of the form $\tilde{\mbb G} = \mc U \mc E \mbb G \mc U^\dagger$.}

For more general gate-dependent noise models, the idea remains more or less the same.
As shown in \cref{sec:appendixB}, the right error $\mc E_R$ is replaced by
its generalization, the $4^{\rm th}$ order right error $\mc E_R^{(4)}=\mbb E\left[{\mc G_{4:1}}^\dagger \tilde{\mc G}_{4:1} \right]$ (\cref{eq:right_error}).
From there, we find:
\begin{prop}[{Finding the appropriate targeted gate-set}]\label{prop:optimization}
	A proper choice of physical basis $\mc U$
	for which \cref{eq:main} applies is
\begin{align}\label{eq:optimization}
\mc U = \argmax\limits_{\mc V} F\left( \mbb E\left[{\mc G_{4:1}}^\dagger \tilde{\mc G}_{4:1} \right] \mc V, \mc I\right)~,
\end{align}
	$\mc U$ cancels the unitary part of the $4^{\rm th}$ order right error.
\end{prop}
This provides
 a means to guide the search
of the appropriate ideal {targeted} gate-set of comparison
$\mc U \mbb G \mc U^\dagger$ given
a numerical noise model $\tilde{\mbb G}$. Indeed,
the $4^{\rm th}$ order right error is easily found, either
by direct computation of the average $\mbb E\left[{\mc G_{4:1}}^\dagger \tilde{\mc G}_{4:1}\right]$,
or more efficiently by solving the eigensystem defined in \cref{eq:left_eigen}.
The optimization defined in \cref{eq:optimization} can be solved
via a gradient ascent parametrized over the $d^2-1$
degrees of freedom of $SU(d)$.

In the single-qubit case, the optimization procedure can be replaced by
an analytical search. Given the process matrix $\bc E_R^{(4)}$ of the $4^{\rm th}$ order right error, it suffices
to find the polar decomposition of its $3 \times 3$ submatrix acting on the
Bloch vectors: $\bc E_R^{(4)}\bm \Pi_{\rm tr}= \bc D_{\rm tr} \bc V_{\rm tr}$.
The unitary factor $\mc V$ corresponds to $\mc U^\dagger$, while the positive factor $\mc D$ captures an incoherent process (rigorously defined in \cref{eq:incoherence}).

With this at hand,
we performed numerically simulated RB experiments
under gate-dependent noise models. Each of the 24 Cliffords
was constructed by a sequence of $X$ and $Y$ pulses,
$G_x=P(\sigma_x,\pi/2)$ and $G_y=P(\sigma_y,\pi/2)$, where
\begin{align}\label{eq:pulse}
	P(H,\theta):= e^{i \theta H/2}~.
\end{align}
 The 2-qubit Cliffords were obtained through the construction
 shown in \cite{Barends2014,Corcoles2013},
 where the 11520 gates are composed of single-qubit Clifford and CZ gates.
The implementation
of the 2-qubit entangling operation was
consistently performed with an over-rotation:
$\tilde{\mc G}_{CZ}=\mc P(\sigma^1_z\sigma^2_z-\sigma_z^1-\sigma^2_z ,\pi/2 + 10^{-1})$. In \cref{fig:pbloch}, the
single-qubit gate generators are modeled
with a slight over-rotation: $\tilde{\mc G}_x=\mc P(\sigma_x,\pi/2+10^{-1})$ and $\tilde{\mc G}_y=\mc P(\sigma_y,\pi/2+10^{-1})$. This model exemplifies
the failure of the maximization
hypothesis proposed in \cref{eq:wrong_hypothesis}.
In \cref{fig:delta,fig:basis_adjustment},
the
single-qubit gate generators are followed
by a short $Z$ pulse, $\tilde{\mc G}_x=\mc P(\sigma_z,\theta_z)\mc G_x$ and $\tilde{\mc G}_y=\mc P(\sigma_z,\theta_z) \mc G_y$,
which reproduces the toy
model used in \cite{Proctor2017}.

\section{Conclusion}

RB experiments estimate the survival probability
decay parameter $p$ of motion-reversal circuits
constituted of operations from a noisy gate-set $\tilde{\mbb G}$ of increasing length (see \cref{eq:intrinsic_RB}).
While motion-reversal is
intrinsic to the experimental RB procedure, the
estimated decay constant $p$ can be interpreted beyond this paradigm.
In this paper we have shown that, in a physically relevant limit, the very same parameter determines an interesting figure of merit, namely the gate-set circuit fidelity (defined in \cref{eq:avg_circuit_fid_construct}):
as a random operation from $\tilde{\mbb G}$ is introduced to a random circuit
constructed from elements in $\tilde{\mbb G}$,
$p$ captures the expected relative change in the gate-set circuit fideilty through \cref{eq:update_fid}.

It is also possible to characterize the full evolution of gate-set circuit fidelity as a function of the
circuit length. In this paper, we have also demonstrated that given a
single-qubit noisy gate-set $\tilde{\mbb G}$ perturbed from $\mbb G$, there
exists {an alternate set of targeted gates obtained through a} physical basis change $\mc U \mbb G \mc U^\dagger$
such that the gate-set circuit fidelity takes the simple form given in \cref{eq:main}. This gives a rigorous underpinning to previous work that has assumed that the experimental RB decay parameter robustly determines a relevant average gate fidelity (\cref{eq:average_circuit_fidelity}) for experimental control under generic gate-dependent scenarios. We conjecture a similar result to hold for higher dimensions and provide numerical evidence and physically motivated arguments to support this conjecture.
Given any specific numerical noise model $\tilde{\mbb G}$ perturbed from $\mbb G$, we showed how to obtain a physical unitary $\mc U$ for which \cref{eq:main} holds. The procedure
can be seen as a fidelity maximation of the $4^{\rm th}$ order right error acting on the gate-set through a unitary correction (see \cref{prop:optimization}).

The introduction of such a physical basis adjustment is natural because it has no effect on how errors accumulate as a function of the sequence length. Rather, it only reflects a basis mismatch to the experimental SPAM procedures. This is in principle detectable by RB experiments but in practice not part of the goals of such diagnostic experiments. In particular, differences in the (independent) basis adjustments  required for distinct  gate-sets will not appear in any characterization of the individual gate-sets, but will be detected when comparing RB experiments for this distinct gate-sets (e.g., comparing dihedral benchmarking and standard randomized benchmarking experiments which have distinct gate-sets but share gates in common, or comparing independent single-qubit RB on two qubits - which has no two-qubit entangling gate - with standard two-qubit RB). We leave the problem of characterizing relative basis mismatch between independent gate-sets as a subject for further work.

\appendix

\section{An expression for the total change in the gate-set circuit fidelity}\label{sec:appendixA}

In this section, we extend the standard RB analysis under gate-dependent noise
provided in \cite{Wallman2017,Merkel2018} in order to prove the claim from
\cref{eq:marginal_law} that standard RB returns the relative variation of the
gate-set circuit fidelity.

Let $\bc A$ be the Liouville matrix of a linear map $\mc A$ and $\Pi_{\rm tr}(\rho) = \rho -\mbb I \tr \rho /d$ be the projector onto the traceless component. {We denote the Frobenius norm, which is defined by the Hilbert-Schmidt inner product, as $\|\cdot\|_F$. For instance, in the qubit case $\| \bm \Pi_{\rm tr} \|_F^2 =3$. We denote the induced $2$-norm as $\|\cdot\|_2$, which corresponds to the maximal singular value.}
Let $e_j$ be the canonical unit vectors, ${A = \sum_{j,k} a_{j,k} e_j e_k^T}$, and
\begin{align}
{\rm vec}(A) = \sum_{j,k} a_{j,k} e_k\otimes e_j~.
\end{align}
Using the identity
\begin{align}\label{eq:vecIdentity}
{\rm vec}(ABC) =(C^T \otimes A) {\rm vec}(B)~,
\end{align}
we have
\begin{align}
	f_{\rm tr} (\tilde{\mbb G} ,\mbb G, m) &= \mbb E \left( {\frac{\inner{ \tilde{\bc G}_{m:1}\bm \Pi_{\rm tr}}{ {\bc G}_{m:1} \bm \Pi_{\rm tr}}}{\|\bm \Pi_{\rm tr}\|_F^2}} \right)~ \notag \\
	&=  \frac{{\rm{vec}}^\dagger(\bm \Pi_{\rm tr})}{{\|\bm \Pi_{\rm tr}\|_F}} \bc T^m \frac{{\rm{vec}}(\bm \Pi_{\rm tr})}{{\|\bm \Pi_{\rm tr}\|_F}}\label{eq:reexpressed_sub_fid}
\end{align}
where the twirling superchannel \cite{Chasseur2015,Proctor2017,Wallman2017} is
\begin{align}
\bc T = \mbb E [\bc G_{\rm tr} \otimes \tilde{\bc G} ]
\end{align}
and $\bc G_{\rm tr} = \bc G \bm \Pi_{\rm tr}$.
Changing  the gate-set $\mbb G$ to $\mc U \mbb G \mc U^\dagger$ for some physical unitary $\mc U$ leaves $\bm \Pi_{\rm tr} =\bc U \bm \Pi_{\rm tr} \bc U^\dagger$. Therefore
\begin{align}
	 f_{\rm tr} (\tilde{\mbb G} ,\mc U \mbb G \mc U^\dagger, m)=
	\frac{{\rm{vec}}^\dagger( \bc U \bm \Pi_{\rm tr})}{{\|\bm \Pi_{\rm tr}\|_F}} \bc T^m \frac{{\rm{vec}}(\bc U \bm \Pi_{\rm tr})}{{\|\bm \Pi_{\rm tr}\|_F}}~.\label{eq:reexpressed_sub_fid_U}
\end{align}
The spectrum of $\bc T$ is unchanged under the basis change
$\bc G \rightarrow \bc U \bc G \bc U^\dagger$. Moreover, its most important eigenvectors are as follows:

\begin{lem}\label{lem:channel_Li}
Let $p$ be the highest eigenvalue of $\bc T$ and
\begin{subequations}
\begin{align}
	\bc A_m &:= p^{-m} \mbb E \left[({\bc G}_{{\rm tr}, m:1})^\dagger
	\bm \Pi_{\rm tr}\tilde{\bc G}_{m:1} \right]~, \label{eq:Ai} \\
	\bc B_m &:= p^{-m} \mbb E \left[\tilde{\bc G}_{m:1} \bm \Pi_{\rm tr}({\bc G}_{{\rm tr},m:1})^\dagger \right]~. \label{eq:Bi}
\end{align}
\end{subequations}
Then we have
\begin{subequations}
\begin{align}
	{\rm{vec}}^\dagger(\bc A_\infty^T) \bc T & = p~ {\rm{vec}}^\dagger(\bc A_\infty^T)~,\label{eq:left_eigen}\\
	\bc T {\rm{vec}}(\bc B_\infty) &= p ~{\rm{vec}}(\bc B_\infty)\label{eq:right_eigen}~.
\end{align}
\end{subequations}
\end{lem}

\begin{proof}
By \cref{eq:vecIdentity},
\begin{align}
	{\rm vec}(\bc B_m)
	&=p^{-m}\mbb E((\bc G_{{\rm tr}, m:1})^* \otimes \tilde{\bc G}_{m:1} ) {\rm vec}(\bm \Pi_{\rm tr})~.
\end{align}
As the Liouville representation is real-valued and the $\bc G_j$ are independent,
\begin{align}
	{\rm vec}(\bc B_m)
	&=\left(\bc T/p\right)^m {\rm vec}(\bm \Pi_{\rm tr})~.
\end{align}
Since the noisy gate-set $\tilde{\mbb G}$ is a small perturbation from $\mbb G$ the spectrum of $\bc T$ will be slightly perturbed from $\{1, 0, 0, \ldots\}$. Therefore $(\bc T/p)^m$ approaches a rank $1$ projector as $m$ increases and so
${\rm vec}(\bc B_\infty)$ is a $+1$-eigenvector of $\bc T/p$.

The same argument applies to $\bc A_\infty^T$.
\end{proof}

\Cref{lem:channel_Li} allows us to write
\begin{align} \label{eq:eigen_decomposition}
	\bc T = p~\frac{{\rm vec}(\bc B_\infty ) {\rm vec}^\dagger( \bc A_\infty^T )}{{\inner{\bc A_\infty^T}{\bc B_\infty}}}+\bm \Delta~,
\end{align}
with $\bm \Delta {\rm vec}\left(\bc B_\infty \right)=  {\rm vec}^\dagger\left( \bc A_\infty^T \right) \bm \Delta =0$.
In \cref{eq:reexpressed_sub_fid_U}, we can expand the vectors as
\begin{subequations}
	\begin{align}
		\frac{{\rm{vec}}^\dagger(\bc U \bm \Pi_{\rm tr})}{{\|\bm \Pi_{\rm tr}\|_F}} &=a(\mc U) \frac{{\rm{vec}}^\dagger( \bc A_\infty ^T)}{\|\bc A_\infty \|_F} + \sqrt{1-a^2(\mc U)} w^\dagger(\mc U) \label{eq:left_vec_expansion} \\
		\frac{{\rm{vec}}(\bc U \bm \Pi_{\rm tr})}{{\|\bm \Pi_{\rm tr}\|_F}} &=b(\mc U) \frac{{\rm{vec}}(\bc B_\infty )}{{\|\bc B_\infty \|_F}} + \sqrt{1-b^2(\mc U)} v(\mc U) \label{eq:right_vec_expansion}
	\end{align}
\end{subequations}
where
\begin{align}
	a(\mc U)&:= {\frac{\inner{\bc A_\infty^T}{\bc U}}{\|\bm \Pi_{\rm tr}\|_F^2}} \left(\frac{\| \bc A_\infty\|_F^2}{{\|\bm \Pi_{\rm tr}\|_F^2}}\right)^{-1/2}~, \label{eq:aiU} \\
	b(\mc U)& := {\frac{\inner{\bc U}{ \bc B_\infty}}{\|\bm \Pi_{\rm tr}\|_F^2}} \left(\frac{\| \bc B_\infty\|_F^2}{{\|\bm \Pi_{\rm tr}\|_F^2}}\right)^{-1/2}~.
\end{align}
and $v(\mc U)$, $w(\mc U)$ are implicitly defined unit vectors.
Using this expansion together with
\cref{eq:eigen_decomposition} in \cref{eq:reexpressed_sub_fid_U}
yields the following result:
\begin{thm}[Total gate-set circuit fidelity] \label{thm:total_change}
	The gate-set circuit fidelity obeys
	\begin{align}\label{eq:total_change}
		\mc F (\tilde{\mbb G} ,\mc U \mbb G \mc U^\dagger, m)
		= \frac{1}{d}+\frac{d-1}{d}\left(C(\mc U) p^m + D(m, \mc U)\right)~,
	\end{align}
	where
	\begin{subequations}
		\begin{align}
			 C(\mc U) :=& {\frac{\inner{\bc A_\infty^T}{\bc U}}{\|\bm \Pi_{\rm tr}\|_F^2}
			\frac{\inner{\bc U}{\bc B_\infty}}{\|\bm \Pi_{\rm tr}\|_F^2}
			\left(\frac{\inner{\bc A_\infty^T}{ \bc B_\infty }}{\|\bm \Pi_{\rm tr}\|_F^2}\right)^{-1}} \notag\\
			 =&{\frac{\inner{\bm \Pi_{\rm tr}}{\bc A_\infty \bc U}}{\|\bm \Pi_{\rm tr}\|_F^2}
				\frac{\inner{\bm \Pi_{\rm tr}}{\bc U^\dagger \bc B_\infty}}{\|\bm \Pi_{\rm tr}\|_F^2}
			\frac{\|\bm \Pi_{\rm tr}\|_F^2}{\inner{\bm \Pi_{\rm tr }}{\bc A_\infty \bc B_\infty }}} \\
			 D(m, \mc U) :=& \sqrt{1-a^2(\mc U)}\sqrt{1-b^2 (\mc U)} w(\mc U)^\dagger \bm \Delta^m v(\mc U) ~. \label{eq:D_i}
		\end{align}
	\end{subequations}
\end{thm}

In \cite{Proctor2017,Wallman2017,Merkel2018} it is shown that
standard RB provides an estimate of $p$.
Notice that $p$ is independent
of the basis in which the ideal gate-set of comparison, $\mc U \mbb G \mc U^\dagger$, is expressed.

From \cref{eq:total_change}, it is straightforward to show that
\begin{align}
	&\delta(m,\mc U) := \frac{ f_{\rm tr}(\tilde{\mbb G}, \mc U \mbb G \mc U^\dagger, m+1)}{ f_{\rm tr}(\tilde{\mbb G}, \mc U \mbb G \mc U^\dagger, m)}- p = \notag \\
	&\sqrt{1-a^2(\mc U)}\sqrt{1-b^2 (\mc U)} \frac{w(\mc U)^\dagger \bm \Delta^{m}(\bm \Delta-p \bm \Pi_{\rm tr}) v(\mc U)}{f_{\rm tr }(\tilde{\mbb G}, \mc U \mbb G \mc U^\dagger, m)}~, \label{eq:rapid_decay}
\end{align}
which is exponentially suppressed. {We show in the next section that the eigenvalues of $\bm \Delta$ are at most of order $\sqrt{1-p}$, which ensures a very fast decay, as shown in \cref{fig:delta}.}
\Cref{eq:marginal_law} is in fact
a reformulation of \cref{eq:rapid_decay}.

\begin{figure*}[ht]
	\centering
	\begin{subfigure}[t]{0.5\textwidth}
		\centering
		\includegraphics[width=0.9\linewidth]{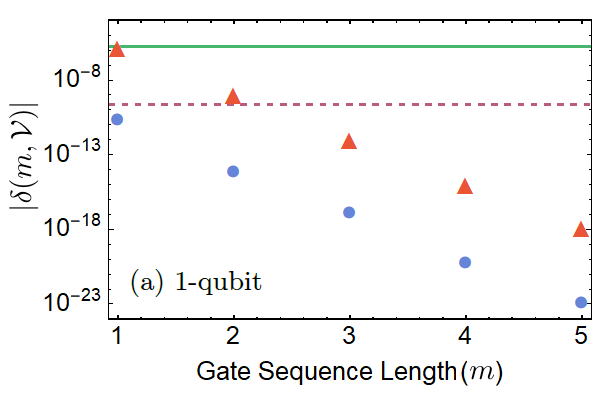}
		\label{fig:delta1}
	\end{subfigure}%
	~
	\begin{subfigure}[t]{0.5\textwidth}
		\centering
		\includegraphics[width=0.9\linewidth]{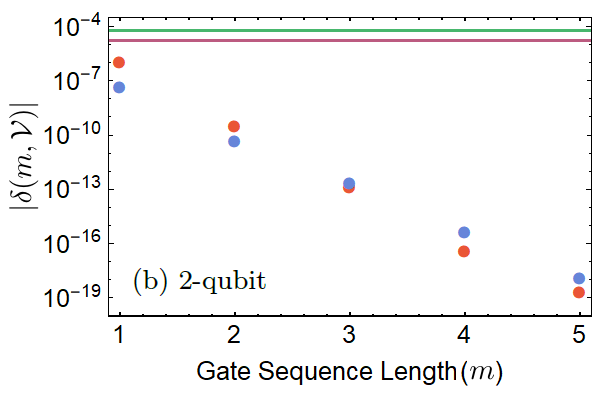}
		\label{fig:delta2}
	\end{subfigure}
	\caption{Absolute value of the deviation $\delta(m, \mc V)$
		, described in \cref{eq:marginal_law} (also see \cref{eq:rapid_decay}), as function of
		circuit length $m$ with noise model generated by
		$\tilde{\mc G}_x=\mc P(\sigma_z,10^{-1})\mc G_x$ and $\tilde{\mc G}_y=\mc P(\sigma_z,10^{-1}) \mc G_y$, $\tilde{\mc G}_{CZ}=\mc P(\sigma^1_z\sigma^2_z-\sigma_z^1-\sigma^2_z ,\pi/2 +10^{-1})$ (see \cref{eq:pulse}). The red {triangles} are obtained with the choice of basis
		$\mc V= \mc I$, while the blue {circles} are obtained with the choice
		$\mc V= \mc U$ where $\mc U$ is found through \cref{eq:optimization}.
		The purple horizontal {dashed} line corresponds to $(1-p)^2$, while
		the {full} green line corresponds to $(1-\mc F( \tilde{\mbb G}, \mbb G, 1 ))^2$.
		For both ideal gate-sets $\mbb G$ and $\mc U \mbb G \mc U^\dagger$,
		the deviation becomes quickly negligible as the sequence length increases.
		In fact, in the case $\mc V=\mc U$ (blue {circles}), the deviation is always below $(1-p)^2$.}
	\label{fig:delta}
\end{figure*}

\section{Varying the ideal gate-set of comparison} \label{sec:appendixB}
\begin{figure*}[ht]
	\centering
	\begin{subfigure}[t]{0.5\textwidth}
		\centering
		\includegraphics[width=0.9\linewidth]{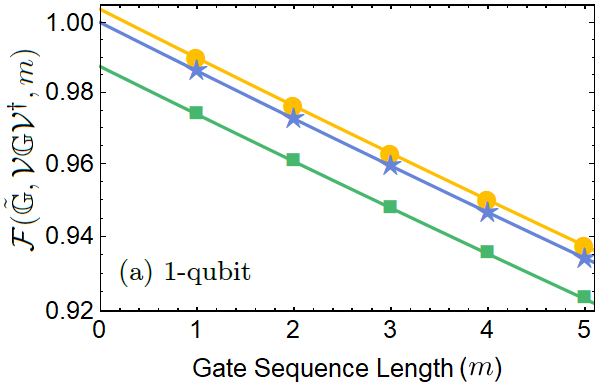}
		\label{fig: pbloch1}
	\end{subfigure}%
	~
	\begin{subfigure}[t]{0.5\textwidth}
		\centering
		\includegraphics[width=0.9\linewidth]{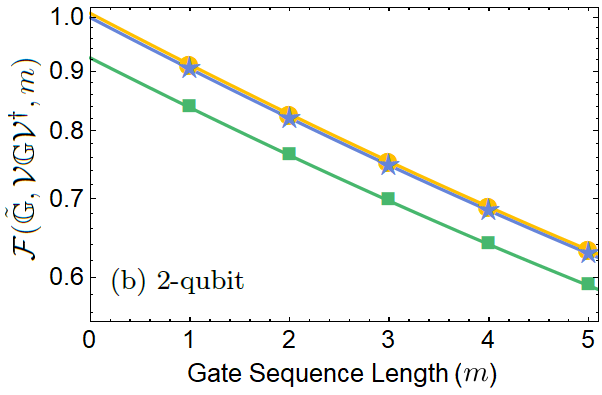}
		\label{fig: pbloch2}
	\end{subfigure}
	\caption{gate-set circuit fidelity $\mc F(\tilde{\mbb G},\mc V \mbb G \mc V^\dagger,m)$ as a function of circuit length $m$ with noise model generated by
		$\tilde{\mc G}_x=\mc P(\sigma_x,\pi/2+10^{-1})$, $\tilde{\mc G}_y=\mc P(\sigma_y,\pi/2+10^{-1})$, $\tilde{\mc G}_{CZ}=\mc P(\sigma^1_z\sigma^2_z-\sigma_z^1-\sigma^2_z ,\pi/2 + 10^{-1})$ (see \cref{eq:pulse}). The different colors portray choices of basis; the yellow circles $\mc V= \mc I$,
		the blue stars $\mc V= \mc U$ where $\mc U$ is found through \cref{eq:optimization},
		and the green squares $\mc V= \mc U^2$. Here the lines correspond to the fit for sequence lengths of m=5 to 10. The choice $\mc V= \mc U$ produces the evolution prescribed by \cref{prop:main}, which through extrapolation has an intercept of $1$.
	}
	\label{fig:pbloch}
\end{figure*}

\begin{figure*}[ht]
	\centering
	\begin{subfigure}[t]{0.5\textwidth}
		\centering
		\includegraphics[width=0.9\linewidth]{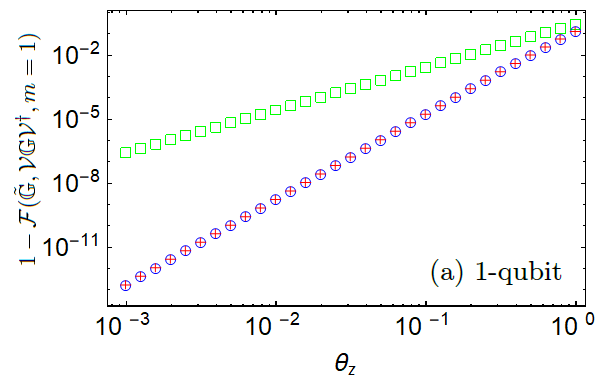}
		\label{fig:basis_adjustment1}
	\end{subfigure}%
	~
	\begin{subfigure}[t]{0.5\textwidth}
		\centering
		\includegraphics[width=0.9\linewidth]{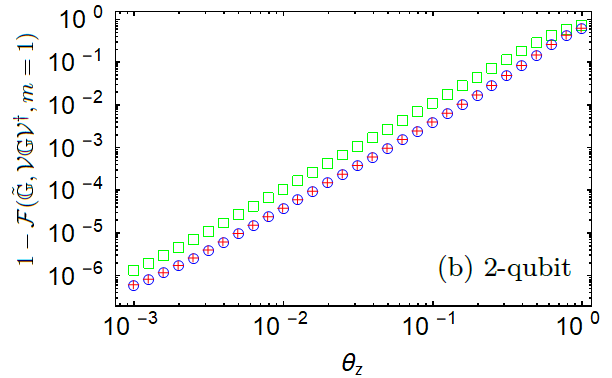}
		\label{fig:basis_adjustment2}
	\end{subfigure}
	\caption{ $1-\mc F(\tilde{\mbb G}, \mc V \mbb G \mc V^\dagger, m=1)$ as function of the angle $\theta_z$ in noise model generated by
		$\tilde{\mc G}_x=\mc P(\sigma_z,\theta_z)\mc G_x$ and $\tilde{\mc G}_y=\mc P(\sigma_z,\theta_z) \mc G_y$, $\tilde{\mc G}_{CZ}=\mc P(\sigma^1_z\sigma^2_z-\sigma_z^1-\sigma^2_z ,\pi/2 +10^{-1})$ (see \cref{eq:pulse}), 
		with $\mc V= \mc I$ (green squares) and $\mc V= \mc U$ (blue circles) where $\mc U$ is found through \cref{eq:optimization}.
		The red crosses correspond to $(1-p)/2$ obtained through RB experiments.}
		\captionsetup{justification   = raggedright,
			singlelinecheck = false, format=hang}
	\label{fig:basis_adjustment}
\end{figure*}

In this section, we prove \cref{prop:main} by determining how the basis $\mc U$
of the ideal gate-set $\mc U \mbb G \mc U^\dagger$ affects the coefficients in
\cref{eq:total_change}.

Let $\mbb G$ be an ideal gate set defined with respect to the SPAM procedures.
We can write the elements of a noisy gate-set as
\begin{align}\label{eq:perturbative_expansion_init}
	\tilde{\bc G}= \bc G + \bm \delta_{\mc G}^{(I)} \bc G~,
\end{align}
so that the perturbations $\delta_{\mc G}$ both capture the errors in the noisy
gate and the mismatch with the targeted computational basis.
Under gate-independent noise with no basis mismatch, $\tilde{\mbb G} = \mc E
\mbb G$ and the infidelity of the perturbed operations $\mc I +\delta_{\mc G}^{(I)}$ is $r(\mc E) := 1 - {\rm F}(\mc E, \mc I)$. 
A basis mismatch will change the infidelity of the perturbations roughly to
$r(\mc U\mc E)+r(\mc U^\dagger)$ for some unitary channel $\mc U$, which will typically differ substantially from the fidelity inferred from the associated RB experiment.

Experimentally, such basis mismatches will be relatively small as operations will be somewhat consistent with SPAM procedures. 
Under this assumption, we now show that there exists an alternate perturbative expansion,
\begin{align}\label{eq:perturbative_expansion_final}
\tilde{\bc G}= \bc U\bc G \bc U^\dagger + \bm \delta_{\mc G}^{(U)} \bc U\bc G \bc U^\dagger~,
\end{align}
for which $r(\mc I+ \mbb E\delta_{\mc G}^{(U)})$ is in line with
the data resulting from an RB experiment.

In \cref{sec:appendixA}, we showed that $(\bc T/p)^n$ converges to a rank-1
projector. We now quantify the rate of convergence. {Recall that $\bc T$ is perturbed from a rank-$1$ projector with spectrum $\{1,0,0, \cdots\}$. Hence,} by the Bauer-Fike
theorem~\cite{Bauer1960}, for any eigenvalue $\lambda\neq p$ of $\bc T$,
\begin{align}
	|\lambda-0| &\leq \| \mbb E [\bc G_{\rm tr} \otimes\bm\delta_{\mc G}^{(I)}\bc G] \|_2 \tag{Bauer-Fike} \\
	&\leq \mbb E \|  [\bc G_{\rm tr} \otimes\bm\delta_{\mc G}^{(I)}\bc G] \|_2 \tag{triangle ineq.} \\
	& = \mbb E \|  \bm\delta_{\mc G}^{(I)} \|_2 \tag{Unitary invariance} \\
	& \leq   O\left(\mbb E  \sqrt{r(\mc I+ \delta_{\mc G}^{(I)})}\right) \tag{\cite{Wallman2015c}} \\
	&\leq  O\left( \sqrt{  r(\mc I+ \mbb E \delta_{\mc G}^{(I)})}\right) \tag{concavity}
\end{align}
This spectral profile implies that $(\bc T/p)^n$ converges quickly to a rank-$1$ operator since the eigenvalues close to zero are exponentially suppressed.

Hence, we can approximate
the asymptotic eigen-operators defined in
\cref{eq:Ai,eq:Bi} as:
	\begin{subequations}
		\label{eq:quick_conv_eigen}
		\begin{align}
		\bc A_\infty &= \bc A_4+O(r(\mc I+ \mbb E\delta_{\mc G}^{(I)})^2)~, \\
		\bc B_\infty &=  \bc B_4+O(r(\mc I+ \mbb E\delta_{\mc G}^{(I)})^2) ~.
		\end{align}
	\end{subequations}
In the simple noise model $\mc E_L \mbb G \mc E_R$, $\bc A_\infty \propto \bm \Pi_{\rm tr} \bc E_R$ and $\bc B_\infty \propto \bc E_L\bm \Pi_{\rm tr}$. To pursue the analogy,
 we denote the $m^{\rm th}$ order right and left errors as
 \begin{subequations}
 	\label{eq:effective_rl}
  		\begin{align}
  		\bc E_R^{(m)}&=~  \mbb E \left[({\bc G}_{ m:1})^\dagger \tilde{\bc G}_{m:1} \right], \label{eq:right_error} \\
  		\bc E_L^{(m)} &= ~\mbb E \left[\tilde{\bc G}_{m:1} ({\bc G}_{m:1})^\dagger \right]~.
  		\end{align}
 \end{subequations}

Combining \cref{eq:effective_rl} and \cref{eq:quick_conv_eigen}, we get
	\begin{subequations}
		\begin{align}\label{eq:A_approx}
		\bc A_\infty &\propto \bm \Pi_{\rm tr }\bc E_R^{(4)} +O(r(\mc I+ \mbb E\delta_{\mc G}^{(I)})^2)~, \\
		\bc B_\infty &\propto  \bc E_L^{(4)} \bm \Pi_{\rm tr } +O(r(\mc I+ \mbb E\delta_{\mc G}^{(I)})^2) ~. \label{eq:B_approx}
 		\end{align}
	\end{subequations}

The structure of
single-qubit error channels allows us to pursue a deeper analysis.
It follows from the channel analysis provided in \cite{Ruskai2002} that,
for high-fidelity qubit-channels, the $3 \times 3$ submatrix acting on the traceless hyperplane can always be decomposed as
\begin{align}\label{eq:decomposition}
{\bc E \bm \Pi_{\rm tr}= \bc D \bc V \bm \Pi_{\rm tr}}
\end{align}
where $\mc V$ is a physical unitary, and $\mc D$ is an incoherent process.
Here we label a channel $\mc D$ incoherent if
\begin{align}\label{eq:incoherence}
{\frac{\inner{\bm \Pi_{\rm tr}}{\bc D}}{\|\bm \Pi_{\rm tr}\|_F^2}= \frac{\|\bc D \bm \Pi_{\rm tr}\|_F}{\|\bm \Pi_{\rm tr}\|_F}}+ O(r(\mc D)^2)~.
\end{align}
{Incoherent channels have the additional property that , given an error channel $\bm \Lambda$ \cite{Dugas2016}
	\begin{align}
		\frac{\inner{\bm \Pi_{\rm tr }}{\bc D \bm \Lambda}}{\|\bm \Pi_{\rm tr }\|_F^2} =
		\frac{\inner{\bm \Pi_{\rm tr }}{\bc D}}{\|\bm \Pi_{\rm tr }\|_F^2}\frac{\inner{\bm \Pi_{\rm tr }}{\bm\Lambda}}{\|\bm \Pi_{\rm tr }\|_F^2}
		+O( r(\mc D\Lambda)^2)~.\label{eq:incoherent_prop}
	\end{align}
}
Expressing the $4^{\rm th}$ order right error  $\bc E_R^{(4)}$ as
	\begin{align}
	{\bc E_{R}^{(4)} \bm \Pi_{\rm tr}=\bc D \bc V \bm \Pi_{\rm tr}~.}
	\end{align}
allows us to maximally correct it through a physical unitary:
\begin{align}
F(\mc E_R^{(4)}\mc V^\dagger, \mc I) = \max\limits_{\mc U} F(\mc E_R^{(4)}\mc U, \mc I) \geq F(\mc E_R^{(4)}, \mc I)~.
\end{align}
Using the property expressed in \cref{eq:incoherent_prop}, we get:
	\begin{align}\label{eq:sub_mult}
 {\frac{\inner{\bm \Pi_{\rm tr}}{\bc E_{R}^{(4)} \bc V^\dagger \bc V \bc E_{L}^{(4)}} }{\|\bm \Pi_{\rm tr}\|_F}}=& {
 \frac{\inner{\bm \Pi_{\rm tr}}{\bc E_{R}^{(4)}  \bc V^\dagger}}{\|\bm \Pi_{\rm tr}\|_F} \frac{\inner{\bm \Pi_{\rm tr}}{\bc V \bc E_{L}^{(4)}} }{\|\bm \Pi_{\rm tr}\|_F}} \notag \\
&+O(r(\mc I + \mbb E \delta_{\mc G}^{(I)})^2).
	\end{align}
Looking back at \cref{thm:total_change} and using \cref{eq:sub_mult,eq:A_approx,eq:B_approx} results in
\begin{align}\label{eq:C_i_corrected}
	C(\mc V^\dagger) =1 +O\left(r(\mc I + \mbb E\delta_{\mc G}^{(I)})^2\right)~.
\end{align}
Since both $\mc V$ and $\mc E_L^{(4)}$ have at most infidelity of order $r(\mc I + \mbb E \delta_{\mc G}^{(I)})$,
it follows that the composition $\mc V \mc E_{L}^{(4)}$ must also have an infidelity
of order $r(\mc I + \mbb E \delta_{\mc G}^{(I)})$, which guarantees
\begin{align}
	\sqrt{1-b^2(\mc V^\dagger)}= O\left(\sqrt{r(\mc I + \mbb E \delta_{\mc G}^{(I)}})\right),
\end{align}
while incoherence guarantees
\begin{align}
\sqrt{1-a^2(\mc V^\dagger)}= O\left(r(\mc I + \mbb E \delta_{\mc G}^{(I)})\right).
\end{align}
Using
\begin{align}
	|w(\mc V^\dagger)^\dagger \bm \Delta v(\mc V^\dagger)| \leq \mbb E \|\delta_{\mc G}^{(I)}\|_2 \leq  O\left(\sqrt{r(\mc I + \mbb E \delta_{\mc G}^{(I)})}\right)
\end{align}
in \cref{eq:D_i}, we find
\begin{align}
 D(1, \mc V^\dagger) = O\left(r(\mc I + \mbb E\delta_{\mc G}^{(I)})^2\right)~,
\end{align}
which, together with \cref{eq:total_change,eq:C_i_corrected} leads to
\begin{align}\label{eq:almost_done}
f_{\rm tr} (\tilde{\mbb G} ,\mc V^\dagger \mbb G \mc V, m) =
p^m + O\left(r(\mc I + \mbb E \delta_{\mc G}^{(I)})^2\right) ~.
\end{align}
This expression allows us to pick a better
perturbative expansion than \cref{eq:perturbative_expansion_init}.
Indeed, choosing
\begin{align}\label{eq:perturbative_expansion_corrected}
\tilde{\bc G}=\bc V^\dagger \bc G \bc V+ \bm \delta^{(V^\dagger)}_{\mc G} ~\bc V^\dagger \bc G \bc V~,
\end{align}
ensures that
the noisy operations $\mc I + \delta^{(V^\dagger)}_{\mc G}$
have an gate-set circuit infidelity which is more in line with
the RB data:
\begin{align}
	r(\mc I + \delta^{(V^\dagger)}_{\mc G})= \frac{d-1}{d}(1-p) + O(r(\mc I + \delta_{\mc G}^{(I)})^2)~.
\end{align}
Iterating the analysis leads to
\begin{align}\label{eq:done}
f_{\rm tr} (\tilde{\mbb G} ,\mc V^\dagger \mbb G \mc V, m) = p^m + O\left((1-p)^2\right) ~.
\end{align}
This completes the demonstration of \cref{prop:main}.

Our current proof technique relies on the structure of single-qubit channels. For higher dimensions, we conjecture that an analog of \cref{prop:main} holds,
although the scaling with the dimension is unclear.

\begin{conj}\label{conj:general}
	If the fidelity of $\mc E_R^{(4)}$ is high,
	then $\exists$ a physical unitary $\mc V^\dagger$ s.t. $\mc E_R^{(4)}\mc V^\dagger$
	is incoherent.
\end{conj}

As we now show constructively, \cref{conj:general} holds for physically
motivated noise models composed of generalized dephasing, amplitude damping,
and unitary processes. Under such noise models,
\begin{align}\label{eq:composite_A}
	\mc E_R^{(4)}= \mc U_T \mc D_T \cdots \mc U_2 \mc D_2 \mc U_1 \mc D_1
\end{align}
for some unitaries $\mc U_i$ and incoherent channels $\mc D_i$.

The channel $\mc U \mc D \mc U^\dagger$ is incoherent for any physical unitary
$\mc U$, and the composition of incoherent channels is also incoherent, so
\cref{eq:composite_A} can be rewritten as $\mc E_R^{(4)}= \mc D \mc V$, where $\mc D$ and $\mc V$
are incoherent and unitary respectively:
\begin{align}
\mc D &= ({\mc U_{T}}\mc D_T {\mc U_{T}}^\dagger) \cdots(\mc U_{T:1} \mc D_1 {\mc U_{T:1}}^\dagger ) \\
\mc V &= \mc U_{T:1}~.
\end{align}

\section*{Acknowledgments}
The authors acknowledge helpful discussions with {Timothy J. Proctor} and {Hui Khoon Ng}. This research was supported by the U.S. Army Research Office through grant W911NF-14-1-0103. This research was undertaken thanks in part to funding from TQT, CIFAR, the Government of Ontario, and the Government of Canada through CFREF, NSERC and Industry Canada.

\bibliography{qcvv}

 \end{document}